\documentclass[lettersize,journal]{IEEEtran}
\usepackage{amsfonts,amsmath,amssymb,amsthm,bm}
\usepackage{booktabs,threeparttable}
\usepackage{graphicx,psfrag,epsf}
\usepackage{enumerate}
\usepackage{url}
\usepackage{algorithm}
\usepackage{algpseudocode}
\usepackage{helvet}
\usepackage{color} 
\usepackage{lineno}
\usepackage{chngcntr}
\usepackage{thmtools, thm-restate}

\theoremstyle{plain}

\newtheorem{lemma}{Lemma}

\newtheorem{definition}{Definition}

\usepackage{adjustbox}
\usepackage{dsfont}
\usepackage{cite}
\newcommand{\blind}{0}
\begin{document}
  \title{Principal Graph Encoder Embedding and Principal Community Detection}
  \author{Cencheng Shen, Yuexiao Dong, Carey E. Priebe, Jonathan Larson, Ha Trinh, Youngser Park
  \thanks{
  \IEEEcompsocthanksitem Cencheng Shen is with the Department of Applied Economics and Statistics, University of Delaware. E-mail: shenc@udel.edu \protect
  \IEEEcompsocthanksitem Yuexiao Dong is with the Department of Statistics, Operations, and Data Science, Temple University. E-mail: ydong@temple.edu \protect
  \IEEEcompsocthanksitem Jonathan Larson and Ha Trinh
are with Microsoft Research at Redmond WA. E-mail: jolarso@microsoft.com, trinhha@microsoft.com \protect
  \IEEEcompsocthanksitem Carey E.Priebe and Youngser Park
are with the Department of Applied Mathematics and Statistics (AMS), the Center for Imaging Science (CIS), and the Mathematical Institute for Data Science (MINDS), Johns Hopkins University. E-mail: cep@jhu.edu, youngser@jhu.edu \protect
}
\thanks{This work was supported in part by 
the National Science Foundation HDR TRIPODS 1934979, 
the National Science Foundation DMS-2113099, and by funding from Microsoft Research.
}}

\maketitle

\begin{abstract}
In this paper, we introduce the concept of principal communities and propose a principal graph encoder embedding method that concurrently detects these communities and achieves vertex embedding. Given a graph adjacency matrix with vertex labels, the method computes a sample community score for each community, ranking them to measure community importance and estimate a set of principal communities. The method then produces a vertex embedding by retaining only the dimensions corresponding to these principal communities. Theoretically, we define the population version of the encoder embedding and the community score based on a random Bernoulli graph distribution. We prove that the population principal graph encoder embedding preserves the conditional density of the vertex labels and that the population community score successfully distinguishes the principal communities. We conduct a variety of simulations to demonstrate the finite-sample accuracy in detecting ground-truth principal communities, as well as the advantages in embedding visualization and subsequent vertex classification. The method is further applied to a set of real-world graphs, showcasing its numerical advantages, including robustness to label noise and computational scalability.
\end{abstract}


\section{Introduction}

\IEEEPARstart{G}{raph} data has become increasingly popular over the past two decades. It plays a pivotal role in modeling relationships between entities across a wide array of domains, including social networks, communication networks, webpage hyperlinks, and biological systems \cite{GirvanNewman2002, newman2003structure, barabasi2004network, boccaletti2006complex, VarchneyEtAl2011, ugander2011anatomy}. Given $n$ vertices and $s$ edges, a binary graph can be represented by an adjacency matrix $\mathbf{A} \in \{0,1\}^{n \times n}$, where $\mathbf{A}(i,j)=1$ means there exists an edge between vertex $i$ and vertex $j$, and $0$ otherwise. The high dimensionality of graph data, dictated by the number of vertices, often necessitates dimension reduction techniques for subsequent inferences.

Dimension reduction techniques applied to graph data are commonly referred to as graph embedding. Specifically, graph embedding transforms the adjacency matrix into a low-dimensional Euclidean representation per vertex. While many such techniques exist, two popular and theoretically sound methods are spectral embedding \cite{Priebe2019} and node2vec \cite{grover2016node2vec}, with asymptotic theoretical guarantees such as convergence to the latent position \cite{SussmanEtAl2012} and consistency in community recovery \cite{zhang2024theoretical}, under popular random graph models such as the stochastic block model and random dot graph model \cite{KarrerNewman2011,ZhaoLevinaZhu2012,JMLR:v18:17-448}. The resulting vertex embeddings facilitate a wide range of downstream inference tasks, such as community detection \cite{RoheEtAl2011, gallagher2023spectral}, vertex classification \cite{TangSussmanPriebe2013, mehta2021neuronal}, and the analysis of multiple graphs and time-series data \cite{arroyo2021inference, Patrick2021}.

The scalability of spectral embedding is often a bottleneck due to its use of singular value decomposition, which can be time-consuming for moderate to large graphs. When vertex labels are available for at least part of the vertex set, a recent method called one-hot graph encoder embedding \cite{GEE1}, which can be viewed as a supervised version of spectral embedding, is significantly faster yet shares similar theoretical properties, such as convergence to the latent positions. It also has several applications to weighted, multiple, and dynamic graphs \cite{GEEFusion, GEEDynamics, GEEDistance, GEERefine, GEECorr}, often exhibiting significantly better finite-sample performance over spectral embedding with a fraction of the time required.

Building upon the one-hot graph encoder embedding, this paper proposes a principal graph encoder embedding algorithm. The key addition is the introduction of a sample community score that ranks the importance of each community. The community score is then used to estimate a set of principal communities that contribute to the decision boundary for separating vertices of different communities. Due to the duality of community and dimensionality in the encoder embedding, the principal graph encoder embedding achieves further dimension reduction by restricting the embedding to the dimensions corresponding to the principal communities. The proposed algorithm maintains the same computational complexity as the original encoder embedding, making it significantly faster than other graph embedding techniques. Additionally, the reduced dimensionality enhances both the speed and robustness of subsequent inference, particularly in the presence of a large number of redundant or noisy communities.

To theoretically justify the sample algorithm, we provide a population characterization of the encoder embedding and principal communities. We prove, under a random Bernoulli graph model, that the principal graph encoder embedding preserves the conditional density of the label vector, making the proposed method Bayes optimal for vertex classification. Furthermore, under a regularity condition, we demonstrate that the proposed sample community score converges to a population community score, which equals zero if and only if the corresponding community is not a principal community.

Through comprehensive simulations and real-data experiments, we validate the numerical performance and theoretical findings through embedding visualization, ground-truth principal community detection, and vertex classification. The proposed method demonstrates excellent numerical accuracy, computational scalability, and robustness against noisy data. Theorem proofs are provided in the appendix. \if0\blind
{The code and data are available on GitHub\footnote{\url{https://github.com/cshen6/GraphEmd}}.} \fi

\section{The Main Method}
\label{sec3}

In this section, we present the principal graph encoder embedding method for a given sample graph, followed by discussions on several practical issues such as normalization, sample community score threshold, and label availability.

\subsection{Principal Graph Encoder Embedding}
\label{sec31}
\begin{itemize}
\item \textbf{Input}: The graph adjacency matrix $\mathbf{A} \in \{0,1\}^{n \times n}$ and a label vector $\mathbf{Y} \in \{0,1,\ldots,K\}^{n}$, where $1$ to $K$ represent known labels, and $0$ is a dummy category for vertices with unknown labels.
\item \textbf{Step 1}: Compute the number of known observations per class, i.e., 
\begin{align*}
n_k = \sum_{i=1}^{n} 1(\mathbf{Y}(i)=k)
\end{align*}
for $k=1,\ldots,K$. 
\item \textbf{Step 2}: Compute the matrix $\mathbf{W} \in [0,1]^{n \times K}$ as follow: for each vertex $i=1,\ldots,n$, set
\begin{align*}
\mathbf{W}(i, k) = 1 / n_k
\end{align*} 
if and only if $\mathbf{Y}(i)=k$, and $0$ otherwise. Note that vertices with unknown labels are effectively assigned zero values, i.e., $\mathbf{W}(i, :)$ is a zero vector if $\mathbf{Y}(i)=0$.
\item \textbf{Step 3}: Compute the original graph encoder embedding through matrix multiplication:
\begin{align*}
\mathbf{Z}=\mathbf{A} \mathbf{W} \in [0,1]^{n \times K}.
\end{align*}
\item \textbf{Step 4 (Normalization)}: Given $\mathbf{Z}$ from step 3, for each $i$ where $\|\mathbf{Z}(i, \cdot)\| > 0$, update the embedding as follows:
\begin{align*}
\mathbf{Z}(i, \cdot) = \frac{\mathbf{Z}(i, \cdot)}{\|\mathbf{Z}(i, \cdot)\|}.
\end{align*}
\item \textbf{Step 5 (Sample Community Score)}: Based on $\mathbf{Z}$ in step 4, 
for each $k \in [1,K]$, compute the sample community score as follows:
\begin{align*}
\hat{\lambda}(k)&= \frac{\max_{l=1,\ldots,K}\{\hat{\mu}(k|l)\} - \min_{l=1,\ldots,K}\{\hat{\mu}(k|l)\}}{\max_{l=1,\ldots,K}\{\hat{\sigma}(\mathbf{\tilde{Z}}(k|l))\}},
\end{align*}
where
\begin{align*}
\hat{\mu}(k|l)&=\frac{\sum_{i=1,\ldots,n}^{\mathbf{Y}(i)=l}\mathbf{Z}(i,k)}{n_l}, \\\hat{\sigma}^2(k|l)&=\frac{\sum_{i=1,\ldots,n}^{\mathbf{Y}(i)=l}\mathbf{Z}^2(i,k)}{n_l-1}-\hat{\mu}^2(k|l),
\end{align*}
then set the estimated principal communities as $\hat{D}=\{k \in [1,K] \mbox{ and } \hat{\lambda}(k)>\epsilon\}$ for a positive threshold $\epsilon$. The choice of $\epsilon$ is discussed in later subsection.
\item \textbf{Step 6 (Principal Encoder)}: Denote the embedding limited to $\hat{D}$ in $\mathbf{Z}$ as $\mathbf{Z}^{\hat{D}}$. Then re-normalize each vertex embedding, i.e., for each $i$, set
\begin{align*}
\mathbf{Z}^{\hat{D}}(i, \cdot)=\frac{\mathbf{Z}(i, \hat{D})}{\|\mathbf{Z}(i, \hat{D})\|},
\end{align*}
\item \textbf{Output}: The original graph encoder embedding $\mathbf{Z}$, the principal graph encoder embedding $\mathbf{Z}^{\hat{D}}$, the sample community score $\{\hat{\lambda}(k)\}$, and the estimated set of principal communities $\hat{D}$. 
\end{itemize}

Note that $\mathbf{Z}^{\hat{D}}$ does not remove any observations from the embedding; rather, it only removes the $k$th dimension when community $k$ is not a principal community, i.e., $\mathbf{Z}^{\hat{D}} \in \mathbb{R}^{n \times |\hat{D}|}$. Every vertex, whether it is from a principal community or not, is always present in the final embedding $\mathbf{Z}^{\hat{D}}$. A brief flowchart of the main method is provided in Figure~\ref{fig0}.

\begin{figure*}[ht]
	\centering
	\includegraphics[width=0.8\textwidth,trim={0cm 0cm 0cm 0cm},clip]{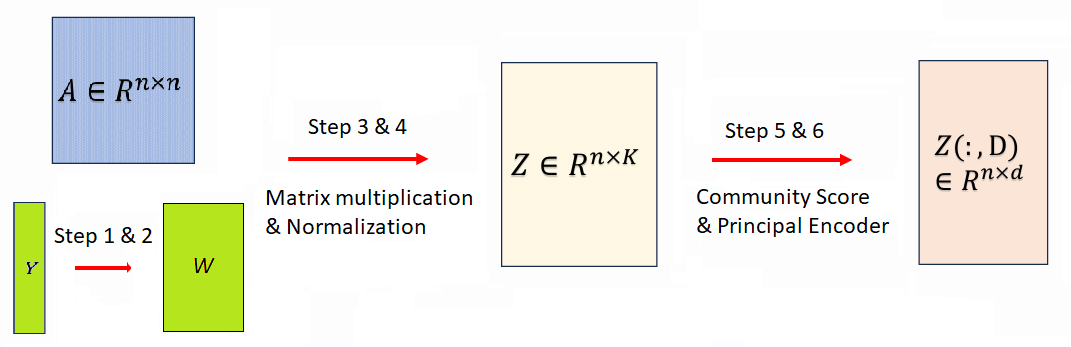}
	\caption{This flowchart illustrates the input, output, and intermediate steps of the principal graph encoder embedding.}
	\label{fig0}
\end{figure*}

\subsection{On Normalization}
Normalization ensures that all vertex embeddings have the same norm. In the context of graph encoder embedding, normalization projects the resulting sample embedding onto a unit sphere, which eliminates degree differences and often leads to improved separation among communities \cite{GEEClustering}. This is particularly beneficial for heterogeneous graphs, which are common in real-world data. In our case, normalization also ensures that the sample community score behaves well. 

\subsection{The Community Score}
The community score is designed to measure the importance of each community. Intuitively, in the original graph encoder embedding, the $k$th dimension can be interpreted as the average connectivity of the target vertex to all vertices in community $k$. Therefore, the proposed community score checks whether there is significant variability within dimension $k$, or equivalently, whether the connectivity from other communities to community $k$ is close to a constant or not. If the connectivity is almost the same, the numerator will be relatively small, or in the extreme case, simply zero, indicating that community $k$ has no information in separating other communities.

A practical question is how large the score should be for a community to qualify as principal. One possible approach is to rank the scores and decide a cut-off via cross-validation. In the presented algorithm, we opt for a faster approach using an adaptive threshold $\epsilon$ for cut-off. To determine this threshold, we employ the profile likelihood method from \cite{ZhuGhodsi2006}, a popular technique for selecting an elbow threshold given a vector. We choose the third elbow of all sample community scores, denoting it as $\epsilon_{\mathbf{A}}$, and then set $\epsilon=\max\{\epsilon_{\mathbf{A}},0.7\}$.

Empirically, the third elbow is very effective for large values of $n$ and $K$. However, for smaller to moderate values of $(n,K)$, the third elbow may be overly conservative. Through experimentation across various models and real datasets, it has been observed that principal communities typically have scores around or higher than 1, while redundant communities tend to have scores of no more than 0.5 for small $(n,K)$. As a result, we settled on the maximum of the third elbow and 0.7 as an empirical choice for $\epsilon$.

\subsection{On Label Vector}
Note that the given algorithm assumes knowledge of the label vector $\mathbf{Y} \in \{0,1,\ldots,K\}^{n}$, which is at least partially known, where $0$ denotes the dummy category of unknown labels. In the absence of a label vector, one could either use a random initialization and k-means to estimate the ground-truth labels \cite{GEEClustering}, or employ a direct label estimation algorithm such as Louvain, Leiden, or label propagation \cite{Louvain2008,Leiden2019,traag2011narrow,raghavan2007near} to estimate a label vector directly from the graph. In either scenario, one can compute the community score and principal encoder accordingly for any estimated label vector. The meaning of principal communities will pertain to the estimated labels, and the population theories in the next section still apply. Therefore, it suffices to assume a given label vector for the purpose of this paper, regardless of whether the label vector is ground-truth or estimated from some algorithms.

\subsection{Computational Complexity}
The computational complexity of the principal graph encoder embedding (P-GEE) is the same as the original graph encoder embedding (GEE), which is $O(nK + s)$, where $s$ represents the number of edges \cite{GEE1,GEEDynamics}. This is because neither the normalization nor the community score computation increases the overall complexity.

For instance, the method is capable of embedding a graph with $100,000$ vertices, $40$ classes, and $10$ million edges in under $10$ seconds on a standard computer using MATLAB code. While the additional steps 5-6 in P-GEE may make it marginally slower than GEE, the reduced dimensionality of $\mathbf{Z}^{\hat{D}}$ can actually enhance its speed and scalability for subsequent tasks such as vertex classification. This is demonstrated in our real data experiments.

\section{Population Definition and Supporting Theory}
In this section, we characterize the population behavior of the method on random graph models. We begin by reviewing several popular random graph models, followed by the introduction of a random graph variable. We then define the population version of the principal community and the graph encoder embedding for this graph variable. This framework allows us to prove that the principal graph encoder embedding preserves the conditional density of the label vector. Additionally, we demonstrate that the sample community score converges to a population community score, which, under a regularity condition, equals zero if and only if the corresponding community is not a principal community. It is important to note that while the other sections focus on the sample method applied to sample graphs, everything in this section pertains to the population version of the method.

\subsection{Existing Random Graph Models}

\subsubsection*{The Stochastic Block Model}
The standard stochastic block model (SBM) is a widely used graph model known for its simplicity and ability to capture community structures \cite{HollandEtAl1983, SnijdersNowicki1997, KarrerNewman2011}. Under SBM, each vertex $i$ is first assigned a class label $\mathbf{Y}(i) \in \{1,\ldots, K\}$. This label can either be predetermined or assumed to follow a categorical distribution with prior probabilities $\{\pi_k \in (0,1], \sum_{k=1}^{K} \pi_k = 1\}$.

Given the vertex labels, the model independently generates each edge between vertex $i$ and another vertex $j \neq i$ using a Bernoulli random variable:
\begin{align*}
\mathbf{A}(i,j) &\sim \mbox{Bernoulli}(B(\mathbf{Y}(i), \mathbf{Y}(j))). 
\end{align*}
Here, $B=[B(k,l)] \in [0,1]^{K \times K}$ represents the block probability matrix, which serves as the parameters of the model. In a directed graph, the lower diagonal of the adjacency matrix is generated using the same distribution, while in an undirected graph, the lower diagonals are set to be equal to the upper diagonals. Note that the model does not have self-loops, meaning that $\mathbf{A}(i,i)=0$. Additionally, whether the graph is directed or undirected does not affect the results here. 

\subsubsection*{The Degree-Corrected Stochastic Block Model}
The standard stochastic block model (SBM) generates dense graphs where all vertices within the same class have the same expected degrees. However, many real-world graphs are heterogeneous, with different vertices having varying degrees, and the graph can be very sparse. To accommodate this, the degree-corrected stochastic block model (DC-SBM) was introduced as an extension of SBM \cite{ZhaoLevinaZhu2012}.

In addition to the existing parameters of SBM, DC-SBM assigns a non-negative and bounded degree parameter $\theta_i$ to each vertex $i$. Given these degrees, the edge between vertex $i$ and another vertex $j \neq i$ is independently generated by:
\begin{align*}
\mathbf{A}(i,j) \sim \mbox{Bernoulli}(\theta_i \theta_j B(\mathbf{Y}(i), \mathbf{Y}(j))).
\end{align*}
When all degrees are set to 1, DC-SBM reduces to the standard SBM. Typically, degrees may be assumed to be fixed a priori or independently and identically distributed within each community. These degree parameters allow DC-SBM to better approximate real-world graphs. 

\subsubsection*{The Random Dot Product Graph}

Under the random dot product graph (RDPG), each vertex $i$ is associated with a hidden latent variable $U_i \stackrel{i.i.d.}{\sim} f_U \in \mathbb{R}^{m}$ \cite{YoungScheinerman2007,JMLR:v18:17-448}. Then each edge is independently generated as follows:
\begin{align*}
\mathbf{A}(i,j) &\sim \mbox{Bernoulli}(<U_i, U_j>),
\end{align*}
where $<\cdot, \cdot>$ denotes the inner product. To enable communities under RDPG, it suffices to assume the latent variable follows a K-component mixture distribution. In other words, each vertex is associated with a class label $\mathbf{Y}(i)$ such that
\begin{align*}
U_i | (\mathbf{Y}(i)=k) &\stackrel{i.i.d.}{\sim} f_{U|k}.
\end{align*}

\subsection{Defining a Graph Variable}

To characterize the graph embedding using a framework similar to the conventional setup of predictor and response variables, we formulate the above graph models into the following graph variable, called the random Bernoulli graph distribution.

Given a vertex, we assume $Y$ is the underlying label that follows a categorical distribution with prior probabilities $\{\pi_k \in (0,1], \sum_{k=1}^{K} \pi_k = 1\}$. Additionally, $X \in \mathbb{R}^{p}$ is the latent variable with a K-component mixture distribution, denoted as
\begin{align*}
X \sim \sum_{k=1}^{K} \pi_k f_{X|Y=k},
\end{align*}
where $f_{X|Y=k}$ represents the conditional density.

Moreover, we assume a known label vector $\mathcal{\vec{V}}=\{v_1, v_2, \ldots, v_m\} \in [1,K]^{m}$, where each $k \in [1,K]$ is present in $\mathcal{\vec{V}}$. Additionally, there exists a corresponding random matrix
\begin{align*}
\mathcal{\vec{U}} &= [U_1;U_2;\cdots;U_m] \in \mathbb{R}^{m \times p},
\end{align*}
where each $U_j$ is independently distributed with density $f_{X|Y=v_j}$.

\begin{definition}
We define an $m$-dimensional random variable $A$ following the random Bernoulli graph distribution as
\begin{align*}
A \sim \mbox{RBG}(X, \mathcal{\vec{V}},\delta) \in \{0,1\}^{m},
\end{align*}
if and only if each dimension $A_j$ is distributed as
\begin{align*}
A_j \sim \mbox{Bernoulli}(\delta(X, U_j)), j=1,\ldots,m.
\end{align*}
Here, $\delta(\cdot, \cdot): \mathbb{R}^p \times \mathbb{R}^p \rightarrow [0,1]$ can be any deterministic function, such as weighted inner product or kernel function. 
\end{definition}
Note that $\mathcal{\vec{V}}$ is a known vector. Alternatively, one could view it as independent sample realizations using the same categorical distribution of $Y$. As we have required each integer from $1$ to $K$ to be present in $\mathcal{\vec{V}}$, it necessarily implies $m \geq K$. 

In essence, the random Bernoulli graph distribution is a multivariate concatenation of mixture Bernoulli distributions. In this distribution, the probability of each Bernoulli trial is determined by a function involving the latent variable $X$ and an independent copy $U_j$ with a known label $v_j$. The random Bernoulli graph distribution is a versatile framework encompassing SBM, DC-SBM, RDPG, and more general cases, due to its flexibility in allowing any $\delta(\cdot, \cdot)$ and any particular distribution for $X$.

Consider the sample adjacency matrix and the labels $(\mathbf{A},\mathbf{Y}) \in \{0,1\}^{n \times n} \times \{1,2,\ldots,K\}^{n}$ as an example, where the graph has no self-loop. Then, for each $i=1,\ldots,n$, the $i$th row of $\mathbf{A}$ is distributed as
\begin{align*}
\mathbf{A}(i,:) \sim \mbox{RBG}(X,\mathcal{\vec{V}},\delta),
\end{align*}
where $X$ is the underlying latent variable for vertex $i$, and $\mathcal{\vec{V}}$ is the known sample labels of all other vertices. Note that the dimension $m=n-1$ because $\mathbf{A}(i,i)=0$, and it suffices to consider the edges between vertex $i$ and all other vertices.

\subsection{The Principal Graph Encoder Embedding for the Graph Variable}
In this section, we characterize the population version of the original graph encoder embedding, the principal community, and the principal graph encoder embedding on the graph variable. Note that their sample notations are $\mathbf{Z}$, $\hat{D}$, and $\mathbf{Z}^{\hat{D}}$ respectively in Section~\ref{sec3}, and the corresponding population notations are $Z$, $D$, and $Z^{D}$ respectively in this section.

\begin{definition}
Given a random graph variable $A \sim \mbox{RBG}(X, \mathcal{\vec{V}},\delta)$. For each $k=1,\ldots,K$, calculate
\begin{align*}
m_k=\sum\limits_{j=1}^{m} 1(v_j=k),
\end{align*}
where $1(v_j=k)$ equals $1$ if $v_j=k$, and $0$ otherwise. 

We then compute the matrix $W \in \mathbb{R}^{m \times K}$ as follows:
\begin{align*}
&W(i,j) = \begin{cases}
1/ m_k     & \quad \text{when $v_j=k$},\\
0  & \quad \text{otherwise.}
\end{cases}
\end{align*}
The population graph encoder embedding is then defined as $Z = A W \in [0,1]^{K}$.
\end{definition}
Note that the $W$ matrix is conceptually similar to the one-hot encoding scheme, except the entries are normalized rather than binary. Next, we introduce the concepts of principal and redundant communities for the graph variable:

\begin{definition}
Given $A \sim \mbox{RBG}(X, \mathcal{\vec{V}}, \delta)$, and $U^{k}$ as an independent variable distributed as $f_{X|Y=k}$. A community $k$ is defined as a principal community if and only if
\begin{align*}
Var\left(E(\delta(X, U^{k}) \mid X)\right) > 0.
\end{align*}
On the other hand, any community for which the above variance equals 0 is referred to as a redundant community.
\end{definition}
For example, in the stochastic block model, the condition $Var(E(\delta(X, U^{k}) \mid X)) = 0$ is equivalent to the $k$th column of the block probability matrix $B(:,k)$ being a constant vector, which does not provide any information about $Y$ via the edge probability. Finally, the principal graph encoder embedding can be defined as follows:
\begin{definition}
Define $D$ as the set of principal communities, and $Z^{D}$ as the graph encoder embedding whose dimensions are restricted to the indices in $D$. We call $Z^{D}$ the principal graph encoder embedding.
\end{definition}
For example, if $K=5$ and $D=\{1,2\}$, then $Z$ spans five dimensions while $Z^{D}$ only keeps the first two dimensions from $Z$. The principal graph encoder embedding achieves additional dimension reduction compared to the original graph encoder embedding. Given the population definitions, the sample versions $\mathbf{Z}$, $\hat{D}$, and $\mathbf{Z}^{\hat{D}}$ in Section~\ref{sec3} can be viewed as sample estimates for the population counterparts $Z$, $D$, and $Z^{D}$.

\subsection{Conditional Density Preserving Property}
Based on the population setting, we can prove the principal graph encoder embedding preserves the conditional density, and as a result, preserves the Bayes optimal classification error via the classical pattern recognition framework \cite{DevroyeGyorfiLugosiBook}.

\begin{restatable}{theorem}{thmOne}
\label{thm1}
Given $A \sim \mbox{RBG}(X, \mathcal{\vec{V}},\delta)$, the principal graph encoder embedding preserves the following conditional density:
\begin{align*}
Y | A \stackrel{dist}{=} Y | Z \stackrel{dist}{=} Y | Z^{D}.
\end{align*}

Denote $L^{*}(Y,A)$ as the Bayes optimal error to classify $Y$ using $A$, we have 
\begin{align*}
L^{*}(Y,A) = L^{*}(Y,Z) = L^{*}(Y,Z^{D}).
\end{align*}
\end{restatable}
Intuitively, the graph variable $A$ is an $m$-dimensional multivariate concatenation of mixture Bernoulli distributions, the original graph encoder embedding $Z$ is a $K$-dimensional multivariate concatenation of mixture Binomial distributions, and the principal graph encoder embedding $Z^{D}$ discards every dimension in $Z$ whose Binomial mixture component is equivalent to a single Binomial. Note that this property is on the population level. For the sample version, we expect the property to hold for sufficiently large vertex size, rather than at any $n$, due to sample estimation variance. 

This theorem shows that the principal communities are well-defined, and retaining the dimensions corresponding to the principal communities is sufficient for subsequent vertex classification. While both the original graph encoder embedding and the principal graph encoder embedding are equivalent in population, the principal graph encoder embedding has fewer dimensions and therefore usually provides a finite-sample advantage in subsequent inference.

\subsection{Population Community Score}
While Theorem~\ref{thm1} establishes that the principal community is well-defined and preserves the conditional density, it remains to demonstrate that the proposed community score can effectively detect such principal communities. To this end, we first introduce the population community score:
\begin{definition}
Define
\begin{align*}
\lambda(k)&= \frac{\max_{l=1,\ldots,K}\{E(Z_k|Y=l)\}}{\max_{l=1,\ldots,K}\{\sqrt{Var(Z_k|Y=l)}\}} \\
&- \frac{\min_{l=1,\ldots,K}\{E(Z_k|Y=l)\}}{\max_{l=1,\ldots,K}\{\sqrt{Var(Z_k|Y=l)}\}} \in [0,+\infty)
\end{align*}
as the population community score for each community $k \in [1,K]$. 
\end{definition}
Since the sample community score used in Step 5 of Section~\ref{sec31} relies on the sample expectation and variance, which converge to their respective population counterparts, it follows immediately that
\begin{align*}
\hat{\lambda}(k) \stackrel{n \rightarrow \infty}{\rightarrow} \lambda(k).
\end{align*} 
In other words, the sample community score converges to the population community score.

The following theorem proves that, under a regularity condition, the population community score perfectly separates principal communities from redundant communities.
\begin{restatable}{theorem}{thmTwo}
\label{thm2}
Assume $A \sim \mbox{RBG}(X, \mathcal{\vec{V}})$, and $\delta(X,U^{k})|Y$ is independent of $X|Y$, which is satisfied under the stochastic block model. Then the population community score $\lambda(k)=0$ if and only if community $k$ is a redundant community. 
\end{restatable}
Note that the condition $\delta(X,U^{k})|Y$ is independent of $X|Y$ can also hold for the degree-corrected stochastic block model, as shown in the proof. Since DC-SBM is often a good model for real-world sparse graphs, the designed community score usually performs well in practice. It is important to emphasize that this property holds at the population level, meaning it is expected to perform well for sufficiently large vertex sizes rather than any finite $n$, due to sample estimation variance.

Finally, there exist other alternative statistics that can consistently detect principal communities. For example, as shown in the proof, one could use the numerator of the community score or the variance of $E(Z_k|Y=l)$, both of which also equal zero under the same condition. Nevertheless, the numerator based on order statistics makes it more robust, and the denominator provides effective normalization for ranking and thresholding purposes, making the proposed community score well-behaved and robust in empirical assessments.

\section{Simulations}
We consider three simulated models with $K=20$ and increasing $n$. In each model, vertex label assignment is randomly determined based on prior probabilities: $\pi_{k}=0.25$ for $k=\{1,2,3\}$, then equally likely to be $0.25/(K-3)$ for the remaining classes. 
Given these labels, edge probabilities are generated under each model. Here are the details of the model parameters for each:

\begin{itemize}
\item \textbf{SBM}: Block probability matrix: $B(k,k)=0.2$ for $k=\{1,2,3\}$, and $B(k,l)=0.1$ otherwise. 
\item \textbf{DC-SBM}: Vertex degree generation:
\begin{align*}
\theta_i | (\mathbf{Y}(i)=y) \sim \mbox{Beta}(1,5+y/5).
\end{align*}
Block probability matrix: $B(1,1)=0.9$, $B(2,2)=0.7$, $B(3,3)=0.5$, and $B(k,l)=0.1$ otherwise. 
\item \textbf{RDPG}: Latent variable $U \in \mathbb{R}^4$. For $k=\{1,2,3\}$, 
\begin{align*}
U(:,k) | (Y=k) \sim \mbox{Uniform}(0.2,0.3). 
\end{align*}
For $k >3$, 
\begin{align*}
U(:,k) | (Y>3) \sim \mbox{Uniform}(0.1,0.2).
\end{align*}
For all dimensions $l \neq k$. 
\begin{align*}
U(:,l) | (Y=k) \sim \mbox{Uniform}(0,0.1).
\end{align*}
\end{itemize}
In all three models, the parameter settings are designed such that vertices from the top three communities can be perfectly separated on a population level, and these communities are considered the principal communities, represented as $D=\{1,2,3\}$. On the other hand, vertices from the remaining communities are intentionally designed to be indistinguishable from each other, constituting redundant communities. As a result, $75\%$ of the vertices belong to the principal communities. 


\subsection{Embedding Visualization and Sample Community Score}
Figure~\ref{fig1} shows the adjacency matrix heatmap for one sample realization, the visualization of the resulting principal graph encoder embedding, and the sample community score for each dimension. The first column shows the heatmap of the generated adjacency matrix $\mathbf{A}$. The sample indices are sorted by class to highlight the block structure. All three graphs exhibit a similar block structure, with higher within-class probabilities for the top three communities. The SBM graph is the most dense graph, followed by RDPG, and the DC-SBM graph is the most sparse by design.

The second column presents the sample community scores ${\hat{\lambda}(k)}$ based on the proposed sample method. Clearly, the sample scores for the first three communities / dimensions stand out and are significantly higher than the others. As a result, the proposed method successfully identifies and reveals the ground-truth dimension, setting $\hat{D}=D=\{1,2,3\}$. 

The third column visualizes the principal graph encoder embedding $\mathbf{Z}^{\hat{D}}$. Each dot represents the embedding for a vertex, and different colors indicate the class membership of each vertex, particularly those from the principal communities. Since $\hat{D}=D=\{1,2,3\}$, the embedding is in 3D and occupies the top three dimensions. We observe that the encoder embedding effectively separates the top three communities, while all redundant communities are mixed together and cannot be distinguished, which aligns with the given models. 


\begin{figure*}[ht]
	\centering	\includegraphics[width=0.9\textwidth,trim={0cm 1cm 0cm 1cm},clip]{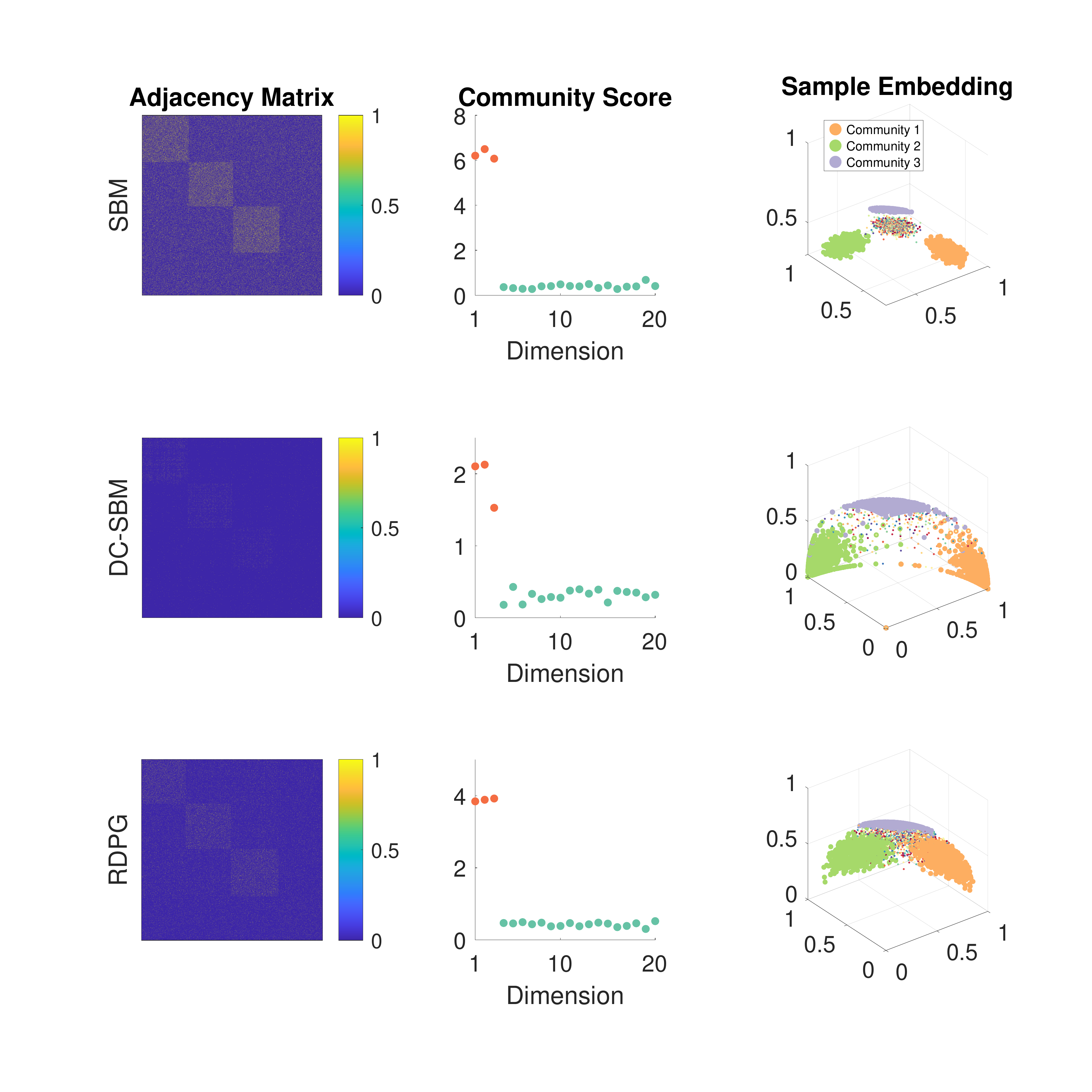}
	\caption{This figure visualizes the sample adjacency matrix for each model at $n=5000$ and $K=20$, the sample community scores for $k \in [1,20]$, and the principal graph encoder embedding.}
	\label{fig1}
\end{figure*}

\subsection{Detection Accuracy and Vertex Classification}
Using the same simulation models, we further assess the capability of the proposed method to identify the ground-truth dimensions and evaluate the quality of the embedding through a classification task on the vertex embedding. For each model, we generate sample graphs with increasing $n$, compute the sample community score, report accuracy in detecting the ground-truth principal communities, compute the principal graph encoder embedding (using training labels only via 5-fold evaluation), apply linear discriminant analysis as the classifier, and report the testing error on the testing vertices. This process is repeated for $100$ Monte-Carlo replicates for each $n$, ensuring that all standard deviations fall within a margin of $1\%$. The average results are reported in Figure~\ref{fig2}.

The first column of Figure~\ref{fig2} shows the sample community scores as $n$ increases. The red line represents the average sample community score among the principal communities, while the blue line represents the average sample community score among the redundant communities. For all three models, as $n$ increases, the principal communities and the redundant communities become increasingly separated. 

This separation translates well to the second column of Figure~\ref{fig2}, which shows that the sample algorithm quickly achieves a true positive rate of $1$ and a false positive rate of $0$ in detecting the true principal communities. In other words, $\text{Prob}(\hat{D}=D) \rightarrow 1$ as $n$ increases. This implies our method is consistent in detecting the ground-truth principal communities.

The third column of Figure~\ref{fig2} evaluates the quality of the graph embedding by conducting a 5-fold vertex classification task on the sample embedding. In each instance, we compute the sample embedding by setting all testing labels to $0$, apply linear discriminant analysis to the embedding and labels of the training vertices, predict the testing labels using the embedding of the testing vertices, and then calculate the error by comparing the predicted label to the true testing label. According to the population model, we can compute that the optimal Bayes error is approximately $0.235$.

As the sample size increases, we observe that the classification error using the principal graph encoder embedding converges to the Bayes-optimal error. The encoder embedding without principal community detection also approaches the Bayes error, though the principal version performs slightly better across all settings.

It is important to note that empirical performance is influenced by both the block structure and the values of $n_k$ within the principal communities. In particular, the DC-SBM simulations are much sparser, requiring larger vertex counts to achieve lower error rates. Additionally, as $n_k$ increases, the classification error decreases accordingly.

\begin{figure*}[ht]
	\centering
	\includegraphics[width=1.0\textwidth,trim={0cm 0cm 0cm 0cm},clip]{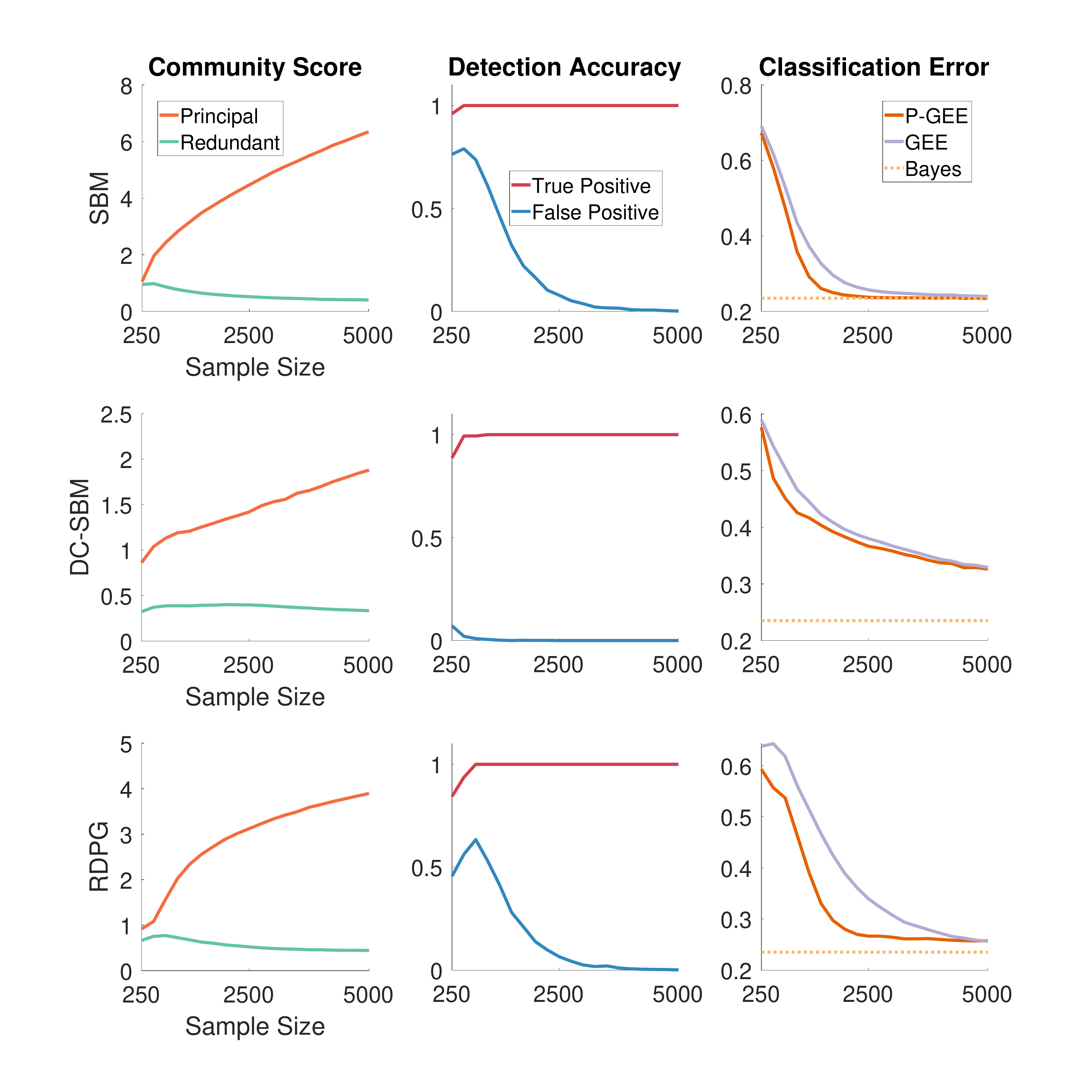}
	\caption{This figure displays the average sample community score, the average principal community detection accuracy, and the average vertex classification error using the embedding, based on $100$ Monte-Carlo replicates with increasing $n$. P-GEE denotes the principal graph encoder embedding, and GEE denotes the original graph encoder embedding.}
	\label{fig2}
\end{figure*}

\section{Real Data}

\subsection{Setting}
We collected a diverse set of real graphs with associated labels from various sources, including the Network Repository\footnote{\url{http://networkrepository.com/}} \cite{nr}, Stanford Network Data\footnote{\url{https://snap.stanford.edu/}}, and internally collected graph data. \if0\blind
{To ensure transparency and reproducibility, we processed all public datasets using MATLAB and have made them available on our GitHub repository.} \fi

Since the ground-truth principal communities are unknown in real graphs, we primarily use vertex classification on embedding as a proxy to evaluate the embedding quality. We compare this to the original graph encoder embedding to assess the quality of the principal community detection. We use 5-fold cross-validation and linear discriminant analysis, and compare the graph encoder embedding (GEE), the principal graph encoder embedding (P-GEE), the adjacency spectral embedding (ASE), and node2vec (N2V). ASE requires an explicit dimension choice, which is set to $d=30$. For node2vec, we use the graspy package \cite{Graspy} with default parameters and $128$ dimensions. Any directed graph was transformed to undirected, and any singleton vertex was removed.

\subsection{Using Original Data}
Table~\ref{table1} summarizes the average error and standard deviation after conducting $100$ Monte Carlo replicates for each given graph. It also provides basic dataset details, including $n$, $K$, and the median dimension choice $|\hat{D}|$ for P-GEE. The numerical results clearly indicate that both GEE and P-GEE deliver excellent performance across all datasets, outperforming spectral embedding in all cases and node2vec in most cases. We also observe that the proposed principal graph encoder embedding is generally very close to the original graph encoder embedding in classification error: by detecting and only retaining the dimensions corresponding to the principal communities, the principal GEE either maintains or slightly improves the classification error compared to the original GEE throughout all data (except the IIP data with $K$ being only $3$). 

This implies that as long as $K$ is not too small, the principal GEE successfully extracts important communities that preserve sufficient label information and improves the classification error. Another observation is that the encoder embedding produces the best error in most cases, and in the two cases where node2vec yields better error than GEE, GEE is also very close in error.

\begin{table*}[ht]
\renewcommand{\arraystretch}{1.3}
\centering
\scalebox{1.0}{
\begin{tabular}{|c|c||c|c|c||c|c|}
 \hline
 & $(n,K)$ & GEE ($\%$) & P-GEE ($\%$) & $|\hat{D}|$ & ASE ($\%$) & N2V ($\%$)\\
\hline
Citeseer  & $(3312,6)$   & $32.8 \pm 0.6$ & $\textbf{32.3} \pm 0.6$ & $4$ & $60.3 \pm 0.5$ & $77.5 \pm 0.5$\\
Cora & $(2708,7)$   & $\textbf{20.9} \pm 1.5$ & $\textbf{20.9} \pm 1.5$ & $5$   & $31.8 \pm 0.6$ & $75.1 \pm 0.5$\\
Email  & $(1005,42)$   & $34.1 \pm 0.8$  & $34.2 \pm 0.8$ & $39$  &  $43.6 \pm 0.4$ & $\textbf{29.2} \pm 0.5$\\
IIP  & $(219,3)$   & $\textbf{31.7} \pm 1.9$ & $32.7 \pm 1.7$ & $2$   & $35.6 \pm 0.4$ & $48.9 \pm 3.2$\\
IMDB  & $(19503,3)$   & $\textbf{1.4} \pm 2.9$ & $\textbf{1.4} \pm 2.9$ & $3$   & $60.1 \pm 0.4$ & $44.8 \pm 0.1$\\
LastFM  & $(7624,18)$   & $20.3 \pm 0.3$ & $20.3 \pm 0.3$ & $17$   & $43.3 \pm 0.4$ & $\textbf{14.7} \pm 0.1$\\
Letter  & $(10507,15)$   & $\textbf{7.4} \pm 0.3$ &  $\textbf{7.4} \pm 0.3$  & $4$   & $89.2 \pm 0.3$ & $74.9 \pm 0.3$\\
Phone & $(1703,71)$   & $30.1 \pm 0.8$ & $\textbf{28.6} \pm 0.8$ & $53$  &  $55.9 \pm 0.2$ & $83.7 \pm 0.5$\\
Protein & $(43471,3)$   & $\textbf{30.8}\pm 0.2$ & $\textbf{30.8}\pm 0.2$ & $3$   & $51.0\pm 0.7$ & $45.8 \pm 0.1$\\
Pubmed & $(19717,3)$   & $\textbf{22.6}\pm 0.2$ & $\textbf{22.6} \pm 0.2$ & $3$   & $35.5\pm 0.7$ & $58.9 \pm 0.2$\\
\hline
\end{tabular}
}
\caption{Vertex classification error using $5$-fold linear discriminant analysis on each graph embedding. The table reports the average error and standard deviation after $100$ Monte Carlo replicates, highlighting the best error within each dataset. All accuracy are in percentile.}
\label{table1}
\end{table*}

\subsection{Using Noisy Data}

To further demonstrate the advantage and robustness of P-GEE, we conducted a noisy data experiment. We used the same real data and evaluation as above, with the exception that the vertex labels were partially polluted. For each replicate, we randomly assigned $10\%$ of the ground-truth vertex labels to one of 30 additional noise classes. For example, suppose $K=3$ and vertex 1 belongs to class 2, so $\mathbf{Y}(1)=2$. If vertex 1 is not polluted, we have $\mathbf{Y}^{noise}(1)=\mathbf{Y}(1)=2$; otherwise, $\mathbf{Y}^{noise}(1) \in [4,5,\ldots,33]$ with equal probability.

We then used the given graph and noisy labels $\mathbf{Y}^{noise}$ to perform vertex embedding and classification for GEE, P-GEE, and ASE, and reported the results in Table~\ref{table2}. Comparing vertex classification accuracy, we observed that P-GEE consistently outperforms GEE and ASE in most cases. In fact, P-GEE is nearly insusceptible to label noise, achieving ideal error rates on the noisy data in most cases. Here, "ideal error" can be defined as the best error on the original data in the corresponding row of Table~\ref{table1}, plus approximately $10\%$. For example, on the IMDB data, the best error on the original data is $1\%$, while P-GEE achieves an error of $10.2\%$ on the noisy data, compared to a much worse error of $28.3\%$ for GEE on the noisy data. Similarly, on the PubMed data, the best error on the original data is $22.6\%$, with P-GEE achieving an error of $32.3\%$ on the noisy data, while GEE on the noisy data has a higher error of $39.0\%$. 

We also observed that the estimated number of principal communities $|\hat{D}|$ accurately matches the true $K$ of the original data in most cases, indicating that the sample community score remains robust and well-behaved in the presence of noisy data.

Finally, we compared the running times of GEE, P-GEE, and ASE, including both embedding and classification time. While the running times for P-GEE and GEE are mostly similar, the reduced dimensionality in P-GEE speeds up subsequent classification, leading to noticeable improvements in running time. Both GEE and P-GEE are significantly faster than spectral embedding for moderate to large graph sizes. For example, for relatively larger graphs, such as LastFM, letter, and protein data, which have tens of thousands of vertices, spectral embedding requires several seconds, whereas GEE only takes a fraction of a second. 

\begin{table*}[ht]
\renewcommand{\arraystretch}{1.3}
\centering
\scalebox{1.0}{
\begin{tabular}{|c||c|c||c|c|c||c|c|}
 \hline
 & GEE ($\%$) & Time (s) & P-GEE ($\%$) & Time (s)& $(K,|\hat{D}|)$ & ASE ($\%$) & Time (s)\\
  \hline
Citeseer   & $49.9 \pm 0.6$ & $0.11$  & $\textbf{41.5} \pm 0.5$ & $0.09$ & $(6,6)$ & $72.9 \pm 0.3$ & $0.16$ \\
 \hline
Cora   & $42.6 \pm 0.6$ & $0.12$   & $\textbf{29.4} \pm 0.5$ & $0.11$ & $(7,7)$ & $61.8 \pm 0.4$ & $0.15$ \\
 \hline
Email & $\textbf{50.5} \pm 0.8$ & $0.38$ & $51.2 \pm 0.9$ & $0.37$ & $(42,60)$ & $50.8 \pm 0.5$ & $0.42$ \\
 \hline
IIP  & $42.8 \pm 2.0$ & $0.02$& $\textbf{39.2} \pm 2.0$ & $0.02$ & $(3,2)$ & $52.3 \pm 2.2$ & $0.03$ \\
\hline
IMDB  & $28.3 \pm 0.4$ & $0.17$ & $\textbf{10.2} \pm 0.05$ & $0.13$ & $(3,3)$ & $69.0 \pm 0.06$ & $0.75$ \\
\hline
LastFM & $42.1 \pm 0.4$ & $0.22$  & $\textbf{30.4} \pm 0.3$ & $0.19$ & $(18,17)$ & $53.0 \pm 0.2$ & $5.9$ \\
\hline
Letter & $29.7 \pm 0.3$  & $0.20$ & $\textbf{21.1} \pm 0.3$ & $0.18$ & $(15,15)$ & $90.8 \pm 0.03$ & $1.1$ \\
\hline
Phone & $48.8 \pm 1.0$  & $0.75$ & $\textbf{41.3} \pm 0.9$ & $0.67$ & $(71,53)$ & $64.1 \pm 0.3$ & $0.70$ \\
\hline
Protein & $42.5 \pm 0.2$ & $0.30$  & $\textbf{38.3} \pm 0.2$ & $0.17$ & $(3,3)$ & $56.6 \pm 0.02$ & $2.4$ \\
\hline
Pubmed& $39.0 \pm 0.3$ & $0.16$  & $\textbf{32.3} \pm 0.2$ & $0.12$ & $(3,3)$ & $44.9 \pm 0.5$ & $0.48$ \\
\hline
\end{tabular}
}
\caption{The evaluation is the same as in Table~\ref{table1}, except that $10\%$ of the given labels are randomized into one of 30 noise groups. We also report the average running time. The best error within the noise columns is highlighted. All accuracy are presented as percentages, and all running times are in seconds.}
\label{table2}
\end{table*}

\section{Conclusion}

In this paper, we propose a Principal Graph Encoder Embedding (P-GEE) algorithm for analyzing a given graph and its associated label vector. P-GEE simultaneously projects the graph into a lower-dimensional space and ranks community importance with respect to the label. We prove that P-GEE preserves conditional density and is Bayes-optimal under the random Bernoulli graph model. Empirically, it demonstrates strong performance in embedding visualization, principal community extraction, vertex classification, and robustness to noisy labels and redundant communities.

Compared to existing methods, such as random walk-based approaches (e.g., node2vec) and neural network-based models like Graph Neural Networks and Graph Autoencoders \cite{kipf2017semi,kipf2016variational,zhou2020graph,Wu2019ACS}, P-GEE offers practical advantages in both interpretability and scalability. Each embedding coordinate $Z(i,k)$ directly represents the edge connectivity between vertex $i$ and community $k$, enabling transparent analysis of community-level influence. This stands in contrast to all other methods, where embedding coordinates lack explicit semantic meaning.

From a computational standpoint, P-GEE is highly efficient due to its non-iterative design. In contrast, spectral methods, node2vec, and Graph Neural Networks rely on iterative procedures, such as singular value decomposition, random walks, or gradient descent, which are computationally expensive and prone to amplifying small noise. We believe the non-iterative structure and explicit exclusion of noisy dimensions make P-GEE more robust, as demonstrated in our noisy data experiments.

While P-GEE performs well across many benchmarks, several future directions remain. First, it is well known that neural network-based methods are sensitive to initialization and may not always perform reliably \cite{oono2020graph,Xu2020What,GCNSpectral}. Simpler approaches such as FastGCN \cite{chen2018fastgcn} and label-based methods \cite{huang2021combining,wang2022combining} have shown competitive performance. In this context, our encoder embedding could serve as a strong initialization scheme or be extended to an iterative, gradient-optimized version, potentially improving performance on complex graph structure, albeit with increased computational cost. Second, modern real-world graphs are increasingly complex. Multiple graphs, multivariate node attributes, overlapping communities, and multiple label vectors are becoming more common. Extending the theoretical framework and algorithmic design of P-GEE to accommodate such scenarios would further enhance its versatility and applicability in real-world settings.

\bibliographystyle{ieeetr}
\bibliography{shen,general}

\begin{IEEEbiography}
[{\includegraphics[width=1in,height=1.25in,clip,keepaspectratio]{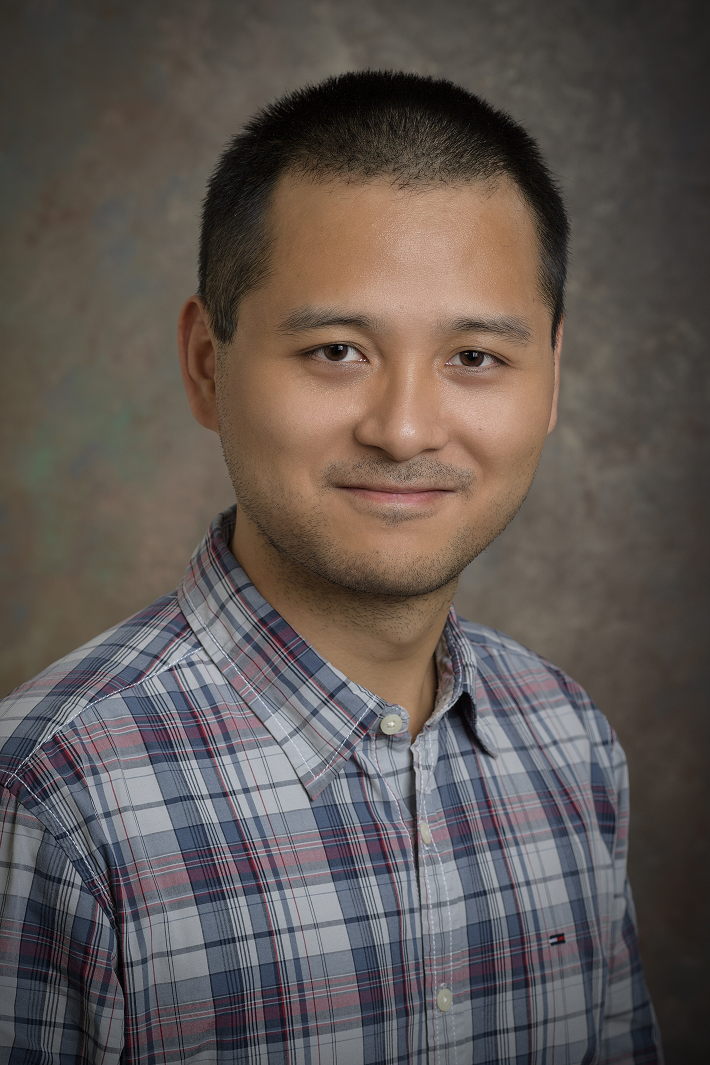}}]{Cencheng Shen received the BS degree in Quantitative Finance from National University of Singapore in 2010, and the PhD degree in Applied Mathematics and Statistics from Johns Hopkins University in 2015. He is an associate professor in the Department of Applied Economics and Statistics at University of Delaware. His research interests include graph embedding, dependence measures, and neural networks.}
\end{IEEEbiography}
\begin{IEEEbiography}
[{\includegraphics[width=1in,height=1.25in,clip,keepaspectratio]{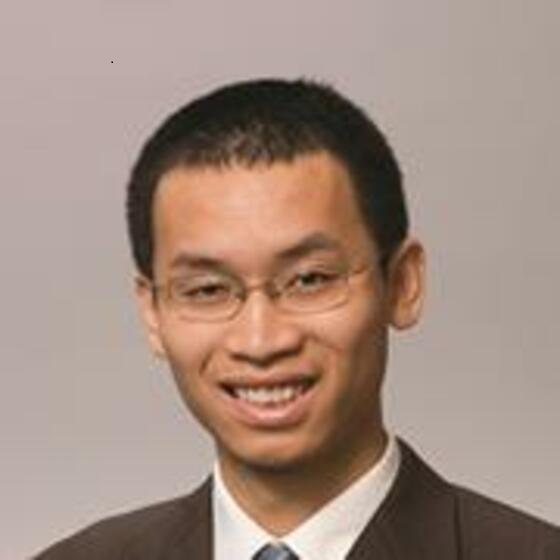}}]{Yuexiao Dong is an Associate Professor from the Department of Statistical Science. Dr. Dong received his Bachelor’s degree in mathematics from Tsinghua University. He obtained his PhD from the statistics department of the Pennsylvania State University in 2009. Dr. Dong’s research focuses on sufficient dimension reduction and high-dimensional data analysis.}
\end{IEEEbiography}
\begin{IEEEbiography}
[{\includegraphics[width=1in,height=1.25in,clip,keepaspectratio]{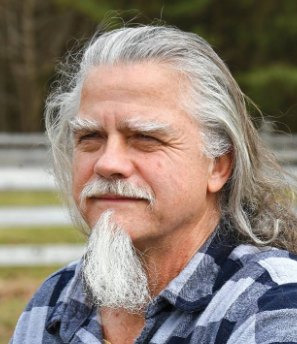}}]{Carey E. Priebe received the BS degree in mathematics from Purdue University in 1984, the MS degree in computer science from San Diego State University in 1988, and the PhD degree in information technology (computational statistics) from George Mason University in 1993. From 1985 to 1994 he worked as a mathematician and scientist in the US Navy research and development laboratory system. Since 1994 he has been a professor in the Department of Applied Mathematics and Statistics at Johns Hopkins University. His research interests include computational statistics, kernel and mixture estimates, statistical pattern recognition, model selection, and statistical inference for high-dimensional and graph data. He is a Senior Member of the IEEE, an Elected Member of the International Statistical Institute, a Fellow of the Institute of Mathematical Statistics, and a Fellow of the American Statistical Association.}
\end{IEEEbiography}
\begin{IEEEbiography}
[{\includegraphics[width=1in,height=1.25in,clip,keepaspectratio]{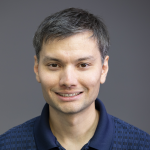}}]{Jonathan Larson is a Partner Data Architect at Microsoft Research working on Special Projects. He currently leads a research team focused on the intersection of graph machine learning, LLM memory representations, and LLM orchestration. Some of his group’s recent research include GraphRAG, organizational science, and cybersecurity. His research has led to shipping new features in Bing, Viva, PowerBI. }
\end{IEEEbiography}
\begin{IEEEbiography}
[{\includegraphics[width=1in,height=1.25in,clip,keepaspectratio]{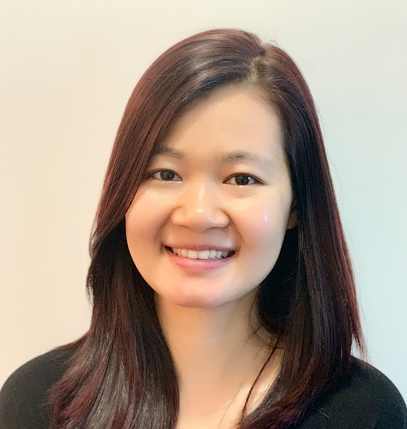}}]{Ha Trinh received a BS degree in Applied Computing from the University of Dundee in 2009 and a PhD degree in Computing from the same university in 2013. She is currently working as a senior data scientist at Microsoft Research. Her research interests lie at the intersection of artificial intelligence and human-computer interaction.}
\end{IEEEbiography}
\begin{IEEEbiography}
[{\includegraphics[width=1in,height=1.25in,clip,keepaspectratio]{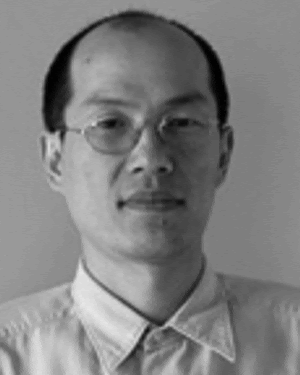}}]{Youngser Park received the B.E. degree in electrical engineering from Inha University in Seoul, Korea in 1985, the M.S. and Ph.D. degrees in computer science from The George Washington University in 1991 and 2011 respectively. From 1998 to 2000 he worked at the Johns Hopkins Medical Institutes as a senior research engineer. From 2003 until 2011 he worked as a senior research analyst, and has been an associate research scientist since 2011 then research scientist since 2019 in the Center for Imaging Science at the Johns Hopkins University. At Johns Hopkins, he holds joint appointments in the The Institute for Computational Medicine and the Human Language Technology Center of Excellence. His current research interests are clustering algorithms, pattern classification, and data mining for high-dimensional and graph data.}
\end{IEEEbiography}

\clearpage
\onecolumn
\setcounter{figure}{0}
\setcounter{theorem}{0}
\renewcommand{\thealgorithm}{C\arabic{algorithm}}
\renewcommand{\thefigure}{E\arabic{figure}}
\renewcommand{\thesection}{A.\arabic{section}}
\renewcommand{\thesubsection}{\thesubsection.\arabic{subsection}}
\renewcommand{\thesubsubsection}{\thesubsection.\arabic{subsubsection}}
\pagenumbering{arabic}
\renewcommand{\thepage}{\arabic{page}}

\bigskip
\begin{center}
{\large\bf APPENDIX}
\end{center}


To facilitate the proof, we introduce the following notations for conditioning and density arguments:
\begin{itemize}
\item We use $\cdot|\mathcal{\vec{U}}$ to denote the conditioning on all the independent variables $U_j$, i.e., for $j=1,2,\ldots,m$, we fix $U_j = u_j$.
\item When conditioning on $(X,Y)=(x,y)$, we simply use $\cdot|(X,Y)$. 
\item We assume $(U,V)$ is an independent copy of $(X,Y)$. Moreover, when conditioning on $(U,V)=(u,v)$, we simply use $\cdot|(U,V)$. 
\item $(a_1, a_2,\ldots, a_m)$ denotes the density argument for each dimension of $A$, and $(z_1, z_2,\ldots, z_K)$ denotes the density argument for each dimension of $Z$. 
\end{itemize}


\thmOne*
\begin{proof}

The proof is decomposed into three parts: 
\begin{itemize}
\item (i) establish $Y | A \stackrel{dist}{=} Y | Z$; 
\item (ii) establish $Y | Z \stackrel{dist}{=} Y | Z^{D}$; 
\item (iii) establish the Bayes error equivalence.\\
\end{itemize}

(i) It suffices to show the following always holds:
\begin{align*}
\text{Prob}(Y | A)= \text{Prob}(Y | Z)
\end{align*}
where $Z = AW$ is the encoder embedding. Given that $Y$ is a categorical variable with prior probabilities $\{\pi_k, k=1,\ldots,K\}$, each conditional probability satisfies
\begin{align*}
\text{Prob}(Y=y | A)&= \frac{\pi_y f_{A|Y=y}(a_1, a_2,\ldots, a_m)}{\sum_{l=1}^{K}\pi_{l}f_{A|Y=l}(a_1, a_2,\ldots, a_m)},\\
\text{Prob}(Y=y | Z)&= \frac{\pi_y f_{Z|Y=y}(z_1, z_2,\ldots, z_K)}{\sum_{l=1}^{K}\pi_{l}f_{Z|Y=l}(z_1, z_2,\ldots, z_K)}.
\end{align*}
Therefore, it suffices to prove that the two numerators are proportional, i.e.,
\begin{align*}
c \times f_{A|Y}(a_1, a_2,\ldots, a_m)= f_{Z|Y}(z_1, z_2,\ldots, z_K)
\end{align*}
for some positive constant $c$ that is unrelated to $Y$.

We begin by examining the conditional density of $A$:
\begin{align*}
f_{A | (Y, X, \mathcal{\vec{U}})}(a_1, a_2,\ldots, a_m) &= \prod_{j=1}^{m} \delta(x,u_j)^{a_j} (1-\delta(x,u_j))^{1-a_j} \\
&= \prod_{k=1}^{K} \prod_{j=1,\ldots,m}^{v_j=k} \delta(x,u_j)^{a_j} (1-\delta(x,u_j))^{1-a_j}.
\end{align*}
The first line follows because each dimension of $A$, under all the conditioning, is independently distributed as a Bernoulli random variable with probability $\delta(x,u_j)$ for $j=1,\ldots,m$. Then the second line rearranges the product based on the class membership of each $v_j$.

We proceed by un-conditioning with respect to $\mathcal{\vec{U}}$, resulting in the following expression:
\begin{align*}
&f_{A | (Y, X)}(a_1,a_2,\ldots,a_m) = \int_{\mathcal{\vec{U}}} f_{A | (Y,X,\mathcal{\vec{U}})}(a_1,\ldots,a_m) f_{\mathcal{\vec{U}}}(u_1,\ldots,u_m)\\
&= \int_{u_1, \ldots, u_m} f_{A | (Y, X, \mathcal{\vec{U}})}(a_1,\ldots,a_m) f_{U_1}(u_1)\cdots f_{U_m}(u_m)\\
&= \int_{u_1, \cdots, u_m} \prod_{k=1}^{K} \prod_{j=1,\ldots,m}^{v_j=k} \delta(x,u_j)^{a_j} (1-\delta(x,u_j))^{1-a_j} f_{U|V=v_1}(u_1)\cdots f_{U|V=v_m}(u_m) \\
& = \prod_{k=1}^{K} \prod_{j=1,\ldots,m}^{v_j=k} E(\delta(x,U)|V=k)^{a_j} (1-E(\delta(x,U)|V=k))^{1-a_j} \\
& = \prod_{k=1}^{K} (\tau_{x,k}(U))^{\sum\limits_{j=1,\ldots,m}^{v_j=k}a_j} (1-\tau_{x,k}(U))^{\sum\limits_{j=1,\ldots,m}^{v_j=k}(1-a_j)}.
\end{align*}
The first line is a standard application of conditional density manipulation, and note that $\mathcal{\vec{U}}$ is independent of $(X,Y)$. The second line rewrites the joint density of $f_{\mathcal{\vec{U}}}$ into individual densities, since the joint density is simply a product of $f_{U|V=v_j}(u_j)$. The fourth line computes the integral: since $u_j$ only appears once in the whole product, either via $\delta(x,u_j)^{a_j}$ or $(1-\delta(x,u_j))^{1-a_j}$ due to $a_j$ taking values of either $0$ or $1$, solving the integral at each $j$ yields either $E(\delta(x,u_j))$ or $(1-E(\delta(x,u_j)))$. Since this expectation is identical throughout the same $v_j$, we can represent this expectation as:
\begin{align*}
\tau_{x,k}(U)=E(\delta(x,U)|V=k).
\end{align*}
This allows us to group terms with the same expectation together based on $k$.

Continuing with the derivation, we un-condition $X$ to derive $f_{A | Y}$:
\begin{align*}
&f_{A | Y}(a_1,a_2,\ldots,a_m) = \int_{x} f_{A | (Y,X)}(a_1,a_2,\ldots,a_m) f_{X|Y}(x)\\
&= \int_{x} \prod_{k=1}^{K} (\tau_{x,k}(U))^{\sum\limits_{j=1,\ldots,m}^{v_j=k}a_j} (1-\tau_{x,k}(U))^{\sum\limits_{j=1,\ldots,m}^{v_j=k}(1-a_j)} f_{X|Y}(x).
\end{align*}

Next, we consider the encoder embedding $Z$. Starting from $f_{Z | (Y, X, \mathcal{\vec{U}})}(z_1, z_2,\ldots, z_K)$, under such conditioning, the density at each dimension $k$ is a Poisson Binomial distribution, i.e.,
\begin{align*}
m_k Z_k | (Y, X, \mathcal{\vec{U}}) \sim \mbox{Poisson Binomial}(\{\delta(x,u_j)\})
\end{align*}
for $j=1,\ldots,m$ and $v_j=k$. After un-conditioning $\mathcal{\vec{U}}$, each probability $\delta(x,u_j)$ again becomes $E(\delta(x,U)|V=k)$ by the same reasoning as above for $f_{A | (Y, X)}$. Therefore,
\begin{align*}
m_k Z_k | (Y, X) \sim \mbox{Binomial}(m_k, \tau_{x,k}(U)).
\end{align*}
As each dimension is conditionally independent, the density of $Z$ is the product of independent Binomials, and we have 
\begin{align*}
f_{Z|(Y, X)}(z_1,z_2,\ldots,z_K)= \prod_{k=1}^{K} {m_k \choose m_k z_k} (\tau_{x,k}(U))^{m_k z_k} (1-\tau_{x,k}(U))^{m_k- m_k z_k},
\end{align*}
and
\begin{align*}
&f_{Z|Y}(z_1,z_2,\ldots,z_K) = \int_{x} f_{Z|(Y, X)}(z_1,z_2,\ldots,z_K) f_{X|Y}(x)\\
&= \int_{x} \prod_{k=1}^{K} {m_k \choose m_k z_k} (\tau_{x,k}(U))^{m_k z_k} (1-\tau_{x,k}(U))^{m_k- m_k z_k} f_{X|Y}(x) \\
&= \prod_{k=1}^{K} {m_k \choose m_k z_k} \int_{x} \prod_{k=1}^{K}(\tau_{x,k}(U))^{m_k z_k} (1-\tau_{x,k}(U))^{m_k- m_k z_k} f_{X|Y}(x),
\end{align*}
where the third line follows because $(m_k,z_k)$ are not affected by the integration of $x$.

Observe that the encoder embedding enforces $m_k z_k=\sum_{j=1,\ldots,m}^{v_j=k}a_j$ for each $k$. Comparing $f_{Z|Y}$ to $f_{A|Y}$, we immediately have 
\begin{align*}
c \times f_{A|Y}(a_1, a_2,\ldots, a_m)= f_{Z|Y}(z_1, z_2,\ldots, z_K)
\end{align*}
when $Z=AW$, where $c=\prod_{k=1}^{K} {m_k \choose m_k z_k}$ is a positive constant.

This conditional density equality holds regardless of the underlying $(X,\mathcal{\vec{V}})$ or $\delta(\cdot, \cdot)$. Hence, we have $Y | A \stackrel{dist}{=} Y | Z$ for the encoder embedding.\\

(ii) Without loss of generality, let us assume that $D=\{1,2,\ldots,d\}$, and $d \in [1,K)$. This means the first $d$ communities are the principal communities, and the remaining are redundant communities. The trivial cases that $d=0$ or $d=K$ will be addressed at the end

Recall the expression from part (i) above:
\begin{align*}
f_{Z|Y}(z_1,z_2,\ldots,z_K) = \prod_{k=1}^{K} {m_k \choose m_k z_k} \int_{x} \prod_{k=1}^{K}(\tau_{x,k}(U))^{m_k z_k} (1-\tau_{x,k}(U))^{m_k- m_k z_k} f_{X|Y}(x).
\end{align*}
This leads to:
\begin{align*}
\text{Prob}(Y=y | Z)&= \frac{\pi_y f_{Z|Y=y}(z_1, z_2,\ldots, z_K)}{\sum_{l=1}^{K}\pi_{l}f_{Z|Y=l}(z_1, z_2,\ldots, z_K)}\\
& =\frac{\pi_y  \prod_{k=1}^{K} {m_k \choose m_k z_k} \int_{x} \prod_{k=1}^{K}(\tau_{x,k}(U))^{m_k z_k} (1-\tau_{x,k}(U))^{m_k- m_k z_k} f_{X|Y=y}(x)}{\sum_{l=1}^{K}\pi_{l} \prod_{k=1}^{K} {m_k \choose m_k z_k} \int_{x} \prod_{k=1}^{K}(\tau_{x,k}(U))^{m_k z_k} (1-\tau_{x,k}(U))^{m_k- m_k z_k} f_{X|Y=l}(x)}\\
& =\frac{\pi_y  \int_{x} \prod_{k=1}^{K}(\tau_{x,k}(U))^{m_k z_k} (1-\tau_{x,k}(U))^{m_k- m_k z_k} f_{X|Y=y}(x)}{\sum_{l=1}^{K}\pi_{l} \int_{x} \prod_{k=1}^{K}(\tau_{x,k}(U))^{m_k z_k} (1-\tau_{x,k}(U))^{m_k- m_k z_k} f_{X|Y=l}(x)}
\end{align*}
where 
\begin{align*}
\tau_{x,k}(U)=E(\delta(x,U^{k})).
\end{align*}

We first look at the terms from community $K$, which is assumed the redundant community. From the definition of redundant community, we have
\begin{align*}
\tau_{x,K}(U)=E(\delta(x,U^{K}))=c_K
\end{align*}
for all possible $x$ where $f_{X}(x)>0$, where $c_K$ is a constant unrelated to $x$. Consequently, all terms involving $\tau_{x,K}(U)$ can be taken outside of the integral in both numerator and denominator, and the same holds for terms associated with $\tau_{x,k}(U)$ for each $k=d+1,\ldots,K$. In essence, for any $l \in [1,K]$, we always have
\begin{align*}
&\pi_l \int_{x} \prod_{k=1}^{K}(\tau_{x,k}(U))^{m_k z_k} (1-\tau_{x,k}(U))^{m_k- m_k z_k} f_{X|Y=l}(x) \\
&=(\prod_{k=d+1}^{K} c_{k}^{m_k z_k}(1-c_{k})^{m_k-m_k z_k}) \pi_l\int_{x} \prod_{k=1}^{d}(\tau_{x,k}(U))^{m_k z_k} (1-\tau_{x,k}(U))^{m_k- m_k z_k} f_{X|Y=l}(x).
\end{align*}
It follows that
\begin{align*}
\text{Prob}(Y=y | Z)=\frac{\pi_y  \int_{x} \prod_{k=1}^{d}(\tau_{x,k}(U))^{m_k z_k} (1-\tau_{x,k}(U))^{m_k- m_k z_k} f_{X|Y=y}(x)}{\sum_{l=1}^{K}\pi_{l} \int_{x} \prod_{k=1}^{d}(\tau_{x,k}(U))^{m_k z_k} (1-\tau_{x,k}(U))^{m_k- m_k z_k} f_{X|Y=l}(x)},
\end{align*}
which exclusively pertains to dimensions corresponding to the principal communities. It is evident that:
\begin{align*}
\text{Prob}(Y=y | Z^{D}) = \text{Prob}(Y=y | Z) =\text{Prob}(Y=y | A) 
\end{align*}
Hence, the principal graph encoder embedding satisfies $Y | A \stackrel{dist}{=} Y | Z^{D}$.

Regarding the two trivial cases: when $d=K$, implying that all communities are principal communities, the theorem trivially holds since no additional dimension reduction occurs. When $d=0$, there is no principal community and $D$ is empty. In this scenario, we have 
\begin{align*}
\text{Prob}(Y=y|A)=\text{Prob}(Y=y|Z) = \pi_y = \text{Prob}(Y=y),
\end{align*}
indicating that $A$ and $Y$ are independent, and $Z$ and $Y$ are independent as well. In other words, the graph provides no information for predicting $Y$, so the graph data itself is redundant.\\

(iii) Given two random variables $(X, Y)$ where $Y$ is categorical, the Bayes optimal classifier for using $X$ to predict $Y$ is
\begin{align*}
g(X) = \arg\max_{k=1,\ldots,K} \text{Prob}(Y = k \mid X).
\end{align*}
By the conditional density equivalence in parts (i) and (ii), it is immediate that the Bayes optimal classifier for using $A$ to predict $Y$ satisfies
\begin{align*}
g(A) &= \arg\max_{k} \text{Prob}(Y = k \mid A) \\
&= \arg\max_{k} \text{Prob}(Y = k \mid Z) = g(Z) \\
&= \arg\max_{k} \text{Prob}(Y = k \mid Z^{D}) = g(Z^{D}).
\end{align*}
Therefore, the Bayes optimal classifier for predicting $Y$ is the same, regardless of whether the underlying random variable is $A$, $Z$, or $Z^{D}$. Since the optimal classifier is always the same, the resulting optimal error is also the same.
\end{proof}

\thmTwo*
\begin{proof}



(i) We first prove that the required condition, $\delta(X,U^{k})|Y$ being independent of $X|Y$, can be satisfied under the stochastic block model (SBM).

Recall that the standard stochastic block model satisfies:
\begin{align*}
\mathbf{A}(i,j) &\sim \mbox{Bernoulli}(B(\mathbf{Y}(i), \mathbf{Y}(j))),
\end{align*}
which, when cast into the framework of a random Bernoulli graph, is equivalent to:
\begin{align*}
\delta(X,U^{k}) | (Y=y) = B(y, k),
\end{align*}
where $B(y, k)$ is a constant and, therefore, always independent of $X|(Y=y)$.

This condition can also hold under the more general degree-corrected stochastic block model (DC-SBM), with an additional assumption regarding how the degree parameters are generated. Under DC-SBM, we have:
\begin{align*}
\delta(X,U^{k}) | (Y=y) = \theta \theta' B(y, k),
\end{align*}
where $\theta$ and $\theta'$ are the degrees for $X$ and $U^{k}$, respectively. Clearly, $\theta' B(y, k)$ is independent of $X|(Y=y)$. By further assuming that the degree variable $\theta$ is generated independently of $X|(Y=y)$, $\delta(X,U^{k})|(Y=y)$ becomes independent of $X|(Y=y)$ under DC-SBM.

(ii) Next, we prove that under the condition that $\delta(X,U^{k})|Y$ is independent of $X|Y$, the population community score $\lambda(k)$ equals $0$ if and only if community $k$ is a redundant community.

From the definition of principal community and the population community score, we need to prove two things. First, we shall prove that
\begin{align*}
Var\left(E(\delta(X, U^{k}) \mid X)\right) = 0
\end{align*}
if and only if
\begin{align*}
Var(E(Z_k|Y=l))=0.
\end{align*}
This is because when the above conditional variance equals 0, $E(Z_k|Y=l)$ is a constant across different $l$, which makes the numerator of the population community score always 0. 

From the proof of Theorem~\ref{thm1}, we have:
\begin{align*}
m_k Z_k | (Y=l, X=x) \sim \mbox{Binomial}(m_k, E(\delta(x,U^{k}))).
\end{align*}
From this, we can derive the conditional expectations as follows:
\begin{align*}
E(Z_k | (Y, X=x)) &= E(\delta(x,U^{k})),\\
E(Z_k | Y=l) &= E(\delta(X,U^{k})|Y=l).
\end{align*}
When community $k$ is redundant such that
\begin{align*}
Var(E(\delta(X, U^{k})|X))=0,
\end{align*}
we immediately have
\begin{align*}
Var(E(Z_k|Y=l))=Var(E(\delta(X,U^{k})|Y=l))=0.
\end{align*}
This is because when the conditional variance equals 0 for all possible $X$, it must also be $0$ when conditioning on $Y=l$, which restricts to part of the support of $X$. This proves the only if direction.

To prove the if direction, we need the additional assumption that $\delta(X,U^{k})|Y$ is independent of $X|Y$. Given such conditional independence, and the fact that $Z_k|Y$ is a random variable with parameter $m_k$ and $E(\delta(X,U^{k}))$, we immediately have that $Z_k|Y$ is independent of $X|Y$, which implies $E(Z_k | Y=l)=E(Z_k | (Y=l, X))$. Therefore, when $Var(E(Z_k | Y=l))=0$, we also have $Var(E(Z_k | (Y=l, X)))=0$. Since $E(Z_k | (Y=l, X=x)) = E(\delta(x,U^{k})) = E(\delta(X,U^{k})|X=x)$, which intuitively means that conditioning on $Y$ is redundant once $X$ is known, this leads to $Var(E(Z_k|X))=Var(E(\delta(X,U^{k})|X))=0$.

(iii) Part (ii) proved that the numerator of the population community score equals $0$ if and only if community $k$ is a redundant community. To complete the proof, it remains to show that the denominator of the population community score is greater than $0$; otherwise, a $0/0$ problem could arise when community $k$ is redundant.

Based on the binomial distribution of $Z_k$, we have:
\begin{align*}
Var(Z_k | (Y, X)) &= \frac{E(\delta(x,U^{k}))(1-E(\delta(x,U^{k})))}{m_k} > 0,
\end{align*}
which always holds, regardless of whether $k$ is redundant or not, except in the trivial case where $\delta(x,U^{k}) = 0$ or 1 almost surely. This corresponds to the scenario where the graph adjacency matrix is entirely 0s or 1s, making all communities redundant and reducing vertex classification to random guessing. Excluding such trivial cases, we have:
\begin{align*}
Var(Z_k | Y=l) &= E_{X}(Var(Z_k | (Y, X)))+Var_{X}(E(Z_k | (Y, X))) > 0,
\end{align*}
where the first term is positive, and the second term is non-negative.

Thus, excluding trivial graphs, the denominator of the principal community score is always positive, regardless of whether community $k$ is redundant or principal.
\end{proof}

\end{document}